\documentclass[preprint,12pt,3p]{elsarticle}

\usepackage{amssymb}

\usepackage{amsthm}
\usepackage{graphicx}
\usepackage{amsmath}
\usepackage{hyperref}

\usepackage{pgf,tikz}

\usetikzlibrary{arrows}
\usetikzlibrary{shapes}
\usepackage{multirow}
\usepackage{latexsym}
\usepackage{setspace}
\PassOptionsToPackage{boxed,section}{algorithm}
\usepackage{enumerate}
\usepackage{amsmath}			
\usepackage{amssymb}			
\usepackage{url}
\usepackage{setspace}
\usepackage{tikz}
\usetikzlibrary{arrows}
\usetikzlibrary{shapes}
\usetikzlibrary{decorations.markings}
\usepackage{geometry}

\newtheorem{proposition}{\em Proposition}
\newtheorem{theorem}{\em Theorem}
\newtheorem{conjecture}{\em Conjecture}
\newtheorem{definition}{\em Definition}

\newtheorem{question}{\em Question}

\journal{Sample Journal}

\begin{document}

\begin{frontmatter}

\title{On sublinear
approximations for the Petersen coloring conjecture}

\author[label1,label2]{Davide Mattiolo}
\fntext[label2]{Partially supported by Fondazione Cariverona, program ``Ricerca Scientifica di Eccellenza 2018'', project ``Reducing complexity in algebra, logic, combinatorics - REDCOM''.}

\ead{davide.mattiolo@univr.it}

\author[label1]{Giuseppe Mazzuoccolo\corref{cor1}}
\address[label1]{Dipartimento di Informatica,
Universit\`a degli Studi di Verona, Strada le Grazie 15, 37134 Verona, Italy}

\cortext[cor1]{Corresponding author}

\ead{giuseppe.mazzuoccolo@univr.it}

\author[label5]{Vahan Mkrtchyan}
\address[label5]{Gran Sasso Science Institute, L'Aquila, 67100, Italy}
\ead{vahan.mkrtchyan@gssi.it}


\begin{abstract}
If $f:\mathbb{N}\rightarrow \mathbb{N}$ is a function, then let us say that $f$ is sublinear if \[\lim_{n\rightarrow +\infty}\frac{f(n)}{n}=0.\] If $G=(V,E)$ is a cubic graph and $c:E\rightarrow \{1,...,k\}$ is a proper $k$-edge-coloring of $G$, then an edge $e=uv$ of $G$ is poor (rich) in $c$, if the edges incident to $u$ and $v$ are colored with three (five) colors. An edge is abnormal if it is neither rich nor poor. The Petersen coloring conjecture of Jaeger states that any bridgeless cubic graph admits a proper 5-edge-coloring $c$, such that there is no an abnormal edge of $G$ with respect to $c$. For a proper 5-edge-coloring $c$ of $G$, let $N_G(c)$ be the set of abnormal edges of $G$ with respect to $c$. In this paper we show that (a) The Petersen coloring conjecture is equivalent to the statement that there is a sublinear function $f:\mathbb{N}\rightarrow \mathbb{N}$, such that all bridgeless cubic graphs admit a proper 5-edge-coloring $c$ with $|N_G(c)|\leq f(|V|)$; (b) for $k=2,3,4$, the statement that there is a sublinear function $f:\mathbb{N}\rightarrow \mathbb{N}$, such that all (cyclically) $k$-edge-connected cubic graphs admit a proper 5-edge-coloring $c$ with $|N_G(c)|\leq f(|V|)$ is equivalent to the statement that all (cyclically) $k$-edge-connected cubic graphs admit a proper 5-edge-coloring $c$ with $|N_G(c)|\leq 2k+1$.
\end{abstract}


\end{frontmatter}

\section{Introduction}
\label{sec:intro}

The Petersen Coloring Conjecture in graph theory asserts that the edges of every bridgeless cubic graph $G$ can be colored with edges of the Petersen graph $P$ as colors, such that adjacent edges of $G$ receive as colors adjacent edges of $P$. The conjecture is considered hard to prove as it implies some other classical conjectures in the field such as Cycle Double Cover Conjecture, Berge-Fulkerson Conjecture, The Shortest Cycle Cover Conjecture (see \cite{Celmins1984,Fulkerson,Jaeger1985,Preiss1981,Zhang1997} for more details). In \cite{Jaeger1985}, Jaeger himself introduced an equivalent formulation of the Petersen Coloring Conjecture. He proved that a bridgeless cubic graph satisfies this conjecture, if and only if, it admits a normal edge-coloring (see Definitions \ref{def:poorrich} and \ref{def:normal}) with at most $5$ colors. Let $\chi'_{N}(G)$ be the minimum number of colors in a normal edge-coloring of $G$. As usual, we refer to $\chi'_{N}(G)$ as the normal chromatic index of $G$. In terms of normal edge-colorings, Petersen Coloring Conjecture amounts to saying that all bridgeless cubic graphs have normal chromatic index at most five. The best known upper bound for $\chi'_{N}(G)$ in the class of all bridgeless cubic graphs is $7$ and some authors have asked for improving this bound to six as an intermediate step towards proving Jaeger's conjecture. If one considers the class of all simple cubic graphs (not necessarily bridgeless), then the situation has been clarified here only recently. There are examples of cubic graphs with normal chromatic index $7$ and it is shown in \cite{JGTpaperNormalColoring} that all simple cubic graphs admit a normal edge-coloring with at most seven colors. 

Now, let us introduce the main notions and definitions that will be used in the paper. Graphs considered in this paper are finite and undirected. They do not contain loops. However, they may contain parallel edges. If $G$ is a graph, then let $V(G)$ and $E(G)$ be the set of vertices and edges of $G$, respectively. For a vertex $v$ of $G$, let $\partial_{G}(v)$ be the set of edges of $G$ that are incident to the vertex $v$ in $G$.

Assume that $G$ and $H$ are two cubic graphs. A mapping $\phi:E(G)\rightarrow E(H)$, such that for each $v\in V(G)$ there is $w\in V(H)$ such that $\phi(\partial_{G}(v)) = \partial_{H}(w)$, is called an $H$-coloring of $G$. If $G$ has an $H$-coloring, then we write $H
\prec G$. It can be easily seen that $\prec$ is a transitive relation defined on the set of cubic graphs. That is, if $H\prec G$ and $K\prec H$, then $K\prec G$.


Let $P$ be the Petersen graph. The Petersen coloring conjecture of Jaeger states:
\begin{conjecture}\label{conj:P10conj} (Jaeger, 1988 \cite{Jaeger1988}) If $G$ is a bridgeless cubic graph, then $P \prec G$.
\end{conjecture}

Note that in \cite{Mkrt2013} it is shown that the Petersen graph is the only 2-edge-connected cubic graph that can color all bridgeless cubic graphs. 
%

For a positive integer $k$ a (proper) $k$-edge-coloring of a graph $G$ is an assignment of colors $\{1,...,k\}$ to edges of $G$, such that adjacent edges receive different colors. If $c$ is an edge-coloring of $G$, then let $S_{G,c}(v)$ be the set of colors that edges incident to the vertex $v$ receive. When it is clear which graph we are referring to, we write $S_{c}(v)$ in place of $S_{G,c}(v)$.

\begin{definition}\label{def:poorrich}
Let $uv$ be an edge of a cubic graph $G$ and $c$ is an edge-coloring of $G$, then the edge $uv$ is called {\bf poor} or {\bf rich} with respect to $c$, if $|S_{c}(u)\cup S_{c}(v)|=3$ or $|S_{c}(u)\cup S_{c}(v)|=5$, respectively. An edge that is neither poor nor rich is called {\bf abnormal}.
\end{definition} For a cubic graph $G$ and an edge-coloring $c$ of $G$, let $N_G(c)$ be the set of abnormal edges of $G$ with respect to $c$.

It can be easily seen that edge-colorings having only poor edges are $3$-edge-colorings of $G$. On the other hand, edge-colorings having only rich edges have been considered in the last years. They are called strong edge-colorings.
In this paper, we consider the case when all edges must be either poor or rich.

\begin{definition}\label{def:normal}
An edge-coloring $c$ of a cubic graph is {\bf normal}, if any edge is rich or poor with respect to $c$. 
\end{definition} 

It is immediate that an edge coloring which assigns a different color to every edge of a simple cubic graph is normal (as all edges are rich). Hence, one can define the normal chromatic index of a simple cubic graph $G$, as the smallest $k$, for which $G$ admits a normal $k$-edge-coloring. Let us denote it by $\chi'_{N}(G)$.


In \cite{Jaeger1985}, Jaeger proved that:

\begin{proposition}\label{prop:JaegerNormalColor}(Jaeger, \cite{Jaeger1985}) If $G$ is a cubic graph, then $P\prec G$, if and only if $G$ admits a normal $5$-edge-coloring.
\end{proposition} This means that Conjecture \ref{conj:P10conj} can be stated as follows:

\begin{conjecture}\label{conj:5NormalConj} For any bridgeless cubic graph $G$, $\chi'_{N}(G)\leq 5$.
\end{conjecture} Observe that Conjecture \ref{conj:5NormalConj} is true for $3$-edge-colorable cubic graphs. This is so because in any $3$-edge-coloring $c$ of a cubic graph $G$ any edge $e$ is poor, hence $c$ is a normal edge-coloring of $G$. This means that non-$3$-edge-colorable cubic graphs are the main obstacle for Conjecture \ref{conj:5NormalConj}. Let us note that Conjecture \ref{conj:5NormalConj} is verified for some non-$3$-edge-colorable bridgeless cubic graphs in \cite{HaggSteff2013}. 

In \cite{Samal2011}, the percentage of edges of a bridgeless cubic graph, which can be made poor or rich in a 5-edge-coloring, is investigated. There it is shown that in any bridgeless cubic graph $G$, there is a proper 5-edge-coloring such that at least $\frac{1}{3}\cdot |E(G)|$ of edges are normal. See the papers \cite{Bilkova15,PartiallyNormal,DAM,Pirot20}, where new results on this problem are presented. The problem studied in \cite{Samal2011} can be viewed as finding proper 5-edge-colorings of bridgeless cubic graphs with some upper bounds for the number of abnormal edges. Note that a similar approach was already considered by Kochol in \cite{Kochol02} for many different problems in graph theory.

The bounds for $\chi'_{N}(G)$ presented in the above-mentioned papers are linear in terms of the size of $G$. Thus, one may wonder whether it could be possible to show that all bridgeless cubic graphs $G$ admit a proper 5-edge-coloring with at most $f(|V(G)|)$ abnormal edges, where $f$ is a fixed sublinear function. 
In this paper, a function $f: \mathbb{N} \rightarrow \mathbb{N}$ on positive integers is called {\bf sublinear}, if 
\[\lim_{n\rightarrow +\infty}\frac{f(n)}{n}=0.\]
In the following section, we prove that obtaining such a result is going to be a difficult task. Here, we would like to offer a more general conjecture:

\begin{conjecture}
\label{conj:PetersenCCAbnormalEdges} The following statements are equivalent:
\begin{enumerate}
    \item [(a)] Petersen coloring conjecture is true;
    
    \item [(b)] There is a sublinear function $f:\mathbb{N}\rightarrow \mathbb{N}$, such that all bridgeless cubic graphs $G$ admit a proper 5-edge-coloring such that at most $f(|V(G)|)$ edges of $G$ are abnormal;
    
    \item [(c)] There is a sublinear function $g:\mathbb{N}\rightarrow \mathbb{N}$, such that all 2-connected cubic graphs $G$ admit a proper 5-edge-coloring such that at most $g(|V(G)|)$ edges of $G$ are abnormal;
    
    \item [(d)] There is a sublinear function $h:\mathbb{N}\rightarrow \mathbb{N}$, such that all 3-connected cubic graphs $G$ admit a proper 5-edge-coloring such that at most $h(|V(G)|)$ edges of $G$ are abnormal;
    
    \item [(e)] There is a sublinear function $k:\mathbb{N}\rightarrow \mathbb{N}$, such that all cyclically 4-edge-connected cubic graphs $G$ admit a proper 5-edge-coloring such that at most $k(|V(G)|)$ edges of $G$ are abnormal.
\end{enumerate}
\end{conjecture} Roughly speaking, our conjecture suggests that beating the linear upper bound for abnormal edges amounts to proving the original conjecture of Jaeger. Note that \cite{Kochol02} proves this statement for many well-known conjectures in graph theory. Table \ref{tab:main} summarizes the results that we have obtained so far. On the left column we give the class of cubic graphs where we assumed the sublinear bound. In the right column we give the constant upper bound that this sublinear bound implies in the class under consideration.

\begin{table}[]
    \centering
    \begin{tabular}{ |c|c| }
 \hline
 The class of cubic graphs & The constant upper bound for abnormal edges \\ 
 \hline
 bridgeless & $= 0$ \\ 
 \hline
 2-edge-connected & $\leq 5$ \\ 
 \hline
 3-edge-connected & $\leq 7$ \\ 
 \hline
 cyclically 4-edge-connected & $\leq 9$ \\ 
 \hline
\end{tabular}
    \caption{The classes of cubic graphs and the upper bounds for abnormal edges that the sublinear bound implies.}\label{tab:main}
\end{table}

If on one hand, upper bounds on the number of abnormal edges can be investigated, on the other hand, one can also try to prove that graphs having a proper $5$-edge-coloring with a small prescribed number of abnormal edges admit a normal $5$-edge-coloring. In this spirit, we present some results and questions in Section~\ref{sec:future_work}. In particular we prove that there is no proper $5$-edge-coloring of a cubic graph having only one abnormal edge.

\section{The main results}

In this section we obtain our main results. Our first theorem says that showing a sublinear bound for the number of abnormal edges with respect to the order of a bridgeless cubic graph $G$ is as hard as proving Petersen coloring conjecture. In other words, we establish the equivalence of statements (a) and (b) in Conjecture \ref{conj:PetersenCCAbnormalEdges}. We observe that in this statement $G$ may not be connected.

\begin{theorem}
\label{thm:bridgeless} The following statements are equivalent:
\begin{enumerate}
    \item [(a)] Conjecture \ref{conj:P10conj} holds true;
    \item [(b)] There exists a sublinear function $f$, such that every bridgeless cubic graph $G$ admits a proper 5-edge-coloring $c$ with $|N_G(c)|\leq f(|V(G)|)$.
\end{enumerate}
\end{theorem}

\begin{proof} (a) implies (b): it follows by Proposition \ref{prop:JaegerNormalColor} since the identically zero function is sublinear.

(b) implies (a):  it is known that it suffices to prove the Petersen coloring conjecture for 2-connected cubic graphs. Let $G$ be such a graph. Consider a cubic graph $H$ obtained from $t$ disjoint copies of $G$. Here $t\geq 1$. We have $|V(H)|=t\cdot |V(G)|$. Since $f$ is sublinear, we can choose $t$ large enough in such a way that
\[\frac{f(|V(H)|)}{|V(H)|}<\frac{1}{|V(G)|}.\]
By (b), we can assume that $H$ admits a proper 5-edge-coloring $c$, such that $|N_H(c)|\leq f(|V(H)|)$. Hence
\[|N_H(c)|\leq f(|V(H)|)<\frac{|V(H)|}{|V(G)|}=t.\]
Thus, there is a copy of $G$ where there are no abnormal edges with respect to $c$. Thus, in this copy $c$ gives a normal 5-edge-coloring of $G$. The proof is complete.
\end{proof}

The proof of previous theorem could appear unsatisfactory since it uses a disconnected graph. Mainly for this reason, we wonder if we can obtain similar conditions under some connectivity assumptions. The answer is positive as we show in the next three theorems.

\begin{theorem}
\label{thm:2edgeconnected} The following statements are equivalent:
\begin{enumerate}
    \item [(a)] Any 2-connected cubic graph $G$ admits a proper 5-edge-coloring $c$, such that $|N_G(c)|\leq 5$.
    \item [(b)] There exists a sublinear function $f$, such that every 2-connected cubic graph $G$ admits a proper 5-edge-coloring $c$ with $|N_G(c)|\leq f(|V(G)|)$.
\end{enumerate}
\end{theorem}

\begin{proof}
(a) implies (b): as before, by Proposition \ref{prop:JaegerNormalColor} and since the identically five function is sublinear.

(b) implies (a): Let $G$ be a 2-connected cubic graph and let $e$ be an edge of $G$. Consider a cubic graph $H$ obtained from $t$ disjoint copies of $G-e$ by joining them cyclically. Here, $t\geq 1$. We have $|V(H)|=t\cdot |V(G)|$. Since $f$ is sublinear, we can choose $t$ such that
\[\frac{f(|V(H)|)}{|V(H)|}<\frac{1}{|V(G)|}.\]
By (b), we can assume that $H$ admits a proper 5-edge-coloring $c$, such that $|N_H(c)|\leq f(|V(H)|)$. Hence
\[|N_H(c)|\leq f(|V(H)|)<\frac{|V(H)|}{|V(G)|}=t.\]
Thus, there is a copy of $G-e$ where there are no abnormal edges with respect to $c$. Since the endpoints of $e$ are adjacent to four edges, they see four colors. Hence there is a color in $\{1,...,5\}$ that is missing in both ends of $e$. We can color $e$ with this color. Clearly, this will give a proper 5-edge-coloring of $G$. Moreover, only $e$ and the four edges adjacent to it might be abnormal. Thus (a) holds. The proof is complete.
\end{proof}

Now we show that we need to increase the number of abnormal edges under stronger connectivity assumptions. 

\begin{theorem}
\label{thm:3edgeconnected} The following statements are equivalent:
\begin{enumerate}
    \item [(a)] Any 3-connected cubic graph $G$ admits a proper 5-edge-coloring $c$, such that $|N_G(c)|\leq 7$.
    \item [(b)] There exists a sublinear function $f$, such that every 3-connected cubic graph $G$ admits a proper 5-edge-coloring $c$ with $|N_G(c)|\leq f(|V(G)|)$.
\end{enumerate}
\end{theorem}

\begin{proof} (a) implies (b): it follows as before by using Proposition \ref{prop:JaegerNormalColor} and the identically seven function is sublinear.

(b) implies (a): let $G$ be any 3-connected cubic graph and let $v$ any of its vertices. For a positive integer $t\geq 1$ take a 3-connected bipartite cubic graph $B$ such that $B$ has $2t$ vertices. Consider a cubic graph $H$ obtained from $B$ by replacing each of its vertices with a copy of $G-v$. Clearly, $H$ is 3-connected and $|V(H)|=2t\cdot (|V(G)|-1)$. Since $f$ is sublinear, we can choose $t$ such that
\[\frac{f(|V(H)|)}{|V(H)|}<\frac{1}{2(|V(G)|-1)}.\]
By (b), there is a proper 5-edge-coloring $c$ of $H$, such that $|N_H(c)|\leq f(|V(H)|)$. We have
\[|N_H(c)|\leq f(|V(H)|)<\frac{|V(H)|}{2(|V(G)|-1)}=t.\]
Thus, since the number of copies of $G-v$ is $t$, there is a copy, say $K$, of $G-v$ which contains no abnormal edges with respect to $c$. 
We construct a copy $G'$ of $G$ by contracting all vertices of $H$ not in $K$ in a unique vertex. We denote again, with a slight abuse of notation, by $v$ such a vertex. 
Let $vv_1$, $vv_2$ and $vv_3$ be the three edges adjacent to $v$ in $G'$. 
We color the edge-set of $G'$ by assign to every edge different from $vv_1$, $vv_2$ and $vv_3$ the same color that it has in $c$. We extend the coloring to the remaining three edges of $G'$ as follows.
We choose the color of $vv_1$ equal to the color of the unique edge of $H$ incident to $v_1$ and not in $K$. Then we can choose the color of $vv_2$ different from the color of $vv_1$ and from the color of the other two edges of $K$ incident to $v_2$. Finally, we can choose the color of $vv_3$ different from the colors of $vv_1$ and $vv_2$ and from the color of the other two edges of $K$ incident to $v_3$. Clearly, this is possible since five colors are avaiable and at most four of them are forbidden in each step of the process. The resulting coloring will be a proper 5-edge-coloring. Moreover, only the edges $vv_1, vv_2, vv_3$ and the four edges adjacent to $vv_2$ and $vv_3$ might be abnormal. Thus, in this coloring there are at most seven abnormal edges. The proof is complete.
\end{proof}

Let $G$ be a bridgeless cubic graph, and let $e_1=ab$ and $e_2=cd$ be two independent edges of it. Take $t$ copies of $G-e_1-e_2$ and let the vertices corresponding to $a, b, c, d$ in the $i$th copy of $G-e_1-e_2$ be $a_i, b_i, c_i, d_i$. Now, join them in the following way: for $i=1,...,(t-1)$ add the edges $d_ia_{i+1}$ and $c_ib_{i+1}$. When $i=t$, we add the edges $d_ta_{1}$ and $c_tb_{1}$ (Figure \ref{fig:JoiningCyclically}). Observe that the resulting graph is cubic. We will say that it is obtained from $t$ copies of $G-e_1-e_2$ by joining them cyclically.

\begin{figure}[ht]
\centering

\includegraphics[scale=0.4]{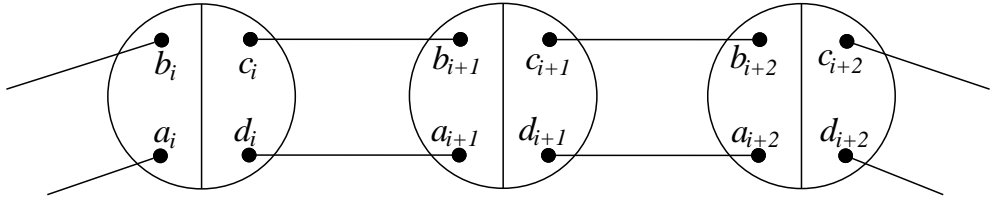}

	\caption{Joining the copies of $G-e_1-e_2$ cyclically.}
	\label{fig:JoiningCyclically}
\end{figure}

We prefer to omit the proof of the following technical proposition. It is nothing but a case by case analysis. 
\begin{proposition}
\label{prop:tcyc4edgeconnected} Let $G$ be a cyclically 4-edge-connected cubic graph and let $e_1$, $e_2$ be two independent edges of $G$. Take $t\geq 2$ copies of $G-e_1-e_2$ and join them cyclically to get a cubic graph $H$. Then $H$ is cyclically 4-edge-connected.
\end{proposition}

Now, by taking a path of length three, such that its end-edges are $e_1$ and $e_2$, we can prove the following

\begin{theorem}
\label{thm:cyc4edgeconnected} The following statements are equivalent:
\begin{enumerate}
    \item [(a)] Any cyclically 4-edge-connected cubic graph $G$ admits a proper 5-edge-coloring $c$, such that $|N_G(c)|\leq 9$.
    \item [(b)] There exists a sublinear function $f$, such that every cyclically 4-edge-connected cubic graph $G$ admits a proper 5-edge-coloring $c$ with $|N_G(c)|\leq f(|V(G)|)$.
\end{enumerate}
\end{theorem}

\begin{proof} (a) implies (b): as before it follows by Proposition \ref{prop:JaegerNormalColor} and since the identically nine function is sublinear.

(b) implies (a): let $G$ be any cyclically 4-edge-connected cubic graph. Take a path of length three in $G$ and let $e_1$, $e_2$ be the end-edges of this path. Now for a positive integer $t\geq 2$, take $t$ copies of $G-e_1-e_2$ and join them cyclically. Let $H$ be the resulting cubic graph. Clearly, $|V(H)|=t\cdot |V(G)|$ and $H$ is cyclically 4-edge-connected (Proposition \ref{prop:tcyc4edgeconnected}). Since $f$ is sublinear, we can choose $t$ such that
\[\frac{f(|V(H)|)}{|V(H)|}<\frac{1}{|V(G)|}.\]
By (b), $H$ admits a proper 5-edge-coloring $c$, such that $|N_H(c)|\leq f(|V(H)|)$. We have:
\[|N_H(c)|\leq f(|V(H)|)<\frac{|V(H)|}{|V(G)|}=t.\]
Thus, there is a copy of $G-e_1-e_2$, such that it contains no abnormal edges with respect to $c$. Now, consider a 5-edge-coloring of $G$ obtained from $c$ by taking the colors of $e_1$ and $e_2$ as the ones missing in its end-points. Since the graph is cubic, there are four such edges. Hence we can choose one from $\{1,...,5\}$. The resulting coloring is proper. Moreover, only the edges $e_1$ and $e_2$ and the edges adjacent to them might be abnormal. Since by our choice $e_1$ and $e_2$ are adjacent to one edge, we have that $G$ contains at most nine abnormal edges with respect to this coloring. The proof is complete.
\end{proof}

\section{Related problems}\label{sec:future_work}

In this paper, we considered the problem stated in \cite{Samal2011}, which can be re-phrased as finding upper bounds for abnormal edges in 5-edge-colorings of bridgeless cubic graphs. Our main goal was to investigate the case of this problem when the upper bound for abnormal edges can be written as a sublinear function of the size of the graph. We presented a conjecture which was implying that obtaining such bounds is going to be as hard as proving the Petersen coloring conjecture. In order to support our conjecture we obtained several results that were showing that sublinear bounds imply constant bounds for abnormal edges. As a problem for future research, one could consider the problem of decreasing these constants. As an intermediate step towards this goal, we would like to offer
\begin{question}
\label{que:Giuseppe} Assume that a bridgeless cubic graph $G$ admits a proper 5-edge-coloring $c$, such that $|N_G(c)|\leq 2$. Can we prove that $G$ admits a normal 5-edge-coloring?
\end{question}


We put constant $2$ in Question \ref{que:Giuseppe} because a cubic graph $G$ with a proper $5$-edge-coloring $c$ is such that $|N_G(c)|\ne 1$. The following result shows this fact.

\begin{figure}
	\centering
	\includegraphics[scale=0.7]{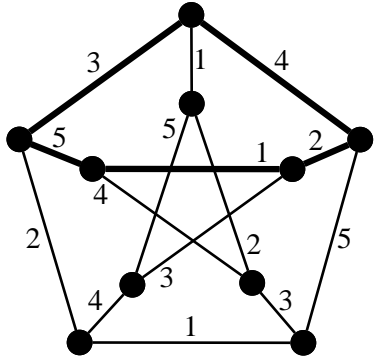}
	\caption{The edge-coloring $\tilde{c}$ of $P.$}\label{Fig:5_edge_coloring_Petersen}
\end{figure}

\begin{proposition}
	There is no cubic graph having a proper $5$-edge-coloring with exactly one abnormal edge.
\end{proposition}

\begin{proof}
	Let $G$ be a cubic graph and $c$ a proper $5$-edge-coloring with the property that exactly one edge $e=uv\in E(G)$ is abnormal. We can assume without loss of generality that $c(e)=1$ and that $S_{G,c}(u)=\{1,2,3\}$, $S_{G,c}(v)=\{1,2,4\}$. Let
	$ P$ be the Petersen graph with the edge-coloring $\tilde{c}$ depicted in Figure \ref{Fig:5_edge_coloring_Petersen}. We denote $\phi_c \colon E(G-e)\to E(P)$ the map such that, for all edges $xy\in E(G-e)$, \begin{itemize}
		\item if $xy$ is poor in $G$, then $\phi_c(xy)=x'y'$ is the unique edge of $P$ such that $c(xy)=\tilde{c}(x'y')$ and $S_{P,\tilde{c}}(x')= S_{G,c}(x) = S_{G,c}(y)$;
		\item if $xy$ is rich in $G$, then $\phi_c(xy)=x'y'$ is the unique edge of $ P$ such that $c(xy)=\tilde{c}(x'y')$, $S_{P,\tilde{c}}(x')=S_{G,c}(x)$ and $S_{P,\tilde{c}}(y') = S_{G,c}(y).$
	\end{itemize}

	Let $C\subset P$ be the bold cycle in Figure \ref{Fig:5_edge_coloring_Petersen} and consider the subgraph $H$ of $G-e$ induced by $\phi^{-1}_c(C)$. First notice that $\phi^{-1}_c(C)\cap (\partial_G(u)\cup \partial_G(v))$ consists of just one edge, that is the edge $wu$, with $w\ne v$, such that $c(wu)=2$. Therefore $H$ has one vertex of degree $1$. On the other hand, since the map $\phi_c$ sends each edge $x_1x_2$ of $G-e$ to an edge $y_1y_2$ of $P$ such that $S_{G,c}(x_i)=S_{P,\tilde{c}}(y_i)$, $i\in\{1,2\}$, we have that every vertex of $V(H)\setminus \{u\}$ has degree $2$ in $H$. Therefore $H$ is a graph with all vertices of degree $2$ and only one vertex of degree $1$, that is impossible.
\end{proof}

On the other hand, we show that, for every integer $k\ge 2$, there is a cubic graph with a proper $5$-edge-coloring having exactly $k$ abnormal edges. Before going to the proof, let $G_1,G_2$ be two cubic graphs and let $e_i=x_iy_i\in E(G_i)$, with $i\in\{1,2\}$. We define the \emph{$2$-cut connection} applied to the graphs $G_1,G_2$ with prescribed edges $e_1,e_2$ respectively, to be the operation that produces the cubic graph $K$ having vertex-set $V(K)=V(G_1)\cup V(G_2)$ and edge-set $E(K)=(E(G_1)\cup E(G_2) \cup \{x_1x_2,y_1y_2\}) \setminus \{e_1,e_2\}$.

\begin{figure}
	\centering
	\includegraphics[scale=0.6]{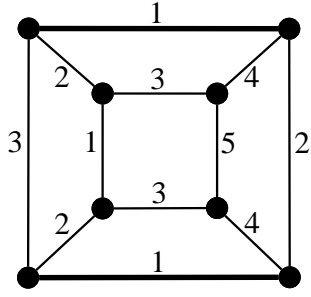}
	\caption{A proper $5$-edge-coloring of the $3$-cube with two abnormal edges (bold edges).}\label{Fig:cube}
\end{figure}

\begin{proposition}
	For all integers $k\ge2$, there is a cubic graph having a proper $5$-edge-coloring with exactly $k$ abnormal edges.
\end{proposition}

\begin{proof}
	Figure \ref{Fig:cube} shows that the statement is true for $k=2$.
	Let $k\ge3$ and $G$ be a cubic graph with a proper $5$-edge-coloring $c$ with $k-1$ abnormal edges. Call $e=uv$ one of the abnormal edges of $G$. Up to a permutation of colors, we can assume that $c(e)=1$ and that $S_{G,c}(u)\cup S_{G,c}(v)=\{1,2,3,4\}$. Construct the cubic graph $H$ applying a $2$-cut connection to $G$ (with prescribed edge $e$) and a copy of $K_4$ (see Figure \ref{Fig:diamond}). The edge-coloring of $G$ can be extended to the new added part as shown in Figure \ref{Fig:diamond}. This new edge-coloring has one more abnormal edge than $c$. Therefore $H$ admits a proper $5$-edge-coloring with $k$ abnormal edges.
\end{proof}

\begin{figure}
	\centering
	\includegraphics[scale=0.6]{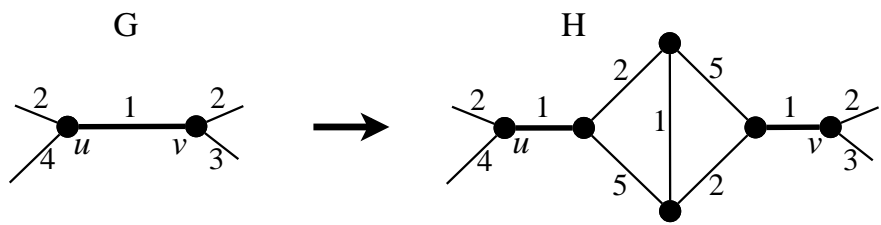}
	\caption{Extendig an edge-coloring in such a way that the resulting one has one more abnormal edge (abnormal edges are depicted in bold).}\label{Fig:diamond}
\end{figure}



\bibliographystyle{elsarticle-num}


\begin{thebibliography}{99}



\bibitem{Bilkova15} H. B\'{i}lkov\'{a}, Variants of Petersen coloring for some graph classes, Master's thesis, Charles University, (2015)

\bibitem{Celmins1984} A. U. Celmins, On cubic graphs that do not have an edge-$3$-colouring, Ph.D. Thesis, Department of Combinatorics and Optimization, University of Waterloo, Waterloo, Canada, 1984.




\bibitem{Fulkerson} D.R. Fulkerson, Blocking and anti-blocking pairs of polyhedra, Math. Programming 1 (1971), 168--194.

\bibitem{HaggSteff2013} J. H\"{a}gglund, E. Steffen, Petersen-colorings and some families of snarks, Ars Mathematica Contemporanea 7 (2014), 161--173.




\bibitem{Jaeger1985} F. Jaeger, On five-edge-colorings of cubic graphs and nowhere-zero flow problems, Ars Combinatoria, 20-B, (1985), 229--244.

\bibitem{Jaeger1988} F. Jaeger, Nowhere-zero flow problems, Selected topics in graph theory, 3, Academic
Press, San Diego, CA, 1988, pp. 71--95.


\bibitem{PartiallyNormal} L. Jin, Y. Kang, Partially normal 5-edge-colorings of cubic graphs, European J. Combin. 95 (2021), 103327.



\bibitem{Kochol02} M. Kochol, Equivalences between hamiltonicity and flow conjectures, and the sublinear defect property, Discr. Math.  254, (2002), 221--230.

\bibitem{JGTpaperNormalColoring} G. Mazzuoccolo, V. V. Mkrtchyan, Normal edge-colorings of cubic graphs, J. Graph Theory 94(1), (2020), 75--91.

\bibitem{DAM} G. Mazzuoccolo, V. V. Mkrtchyan, Normal $6$-edge-colorings of some bridgeless cubic graphs, Disc. Appl. Math. 277, (2020), 252--262.

\bibitem{Mkrt2013} V. Mkrtchyan, A remark on the Petersen coloring conjecture of Jaeger, Australasian J. Comb. 56(2013), pp. 145--151.



\bibitem{Pirot20} F. Pirot, J.-S. Sereni, R. \v{S}krekovski, Variations on the Petersen colouring conjecture, Electron. J. Combin. 27(1) (2020), \#P1.8

\bibitem{Preiss1981} M. Preissmann, Sur les colorations des aretes des graphes cubiques, These de $3$-eme cycle, Grenoble (1981).

\bibitem{Samal2011} R. \v{S}\'{a}mal, New approach to Petersen coloring, Elec. Notes in Discr. Math. 38 (2011), 755--760.





\bibitem{Zhang1997} C.-Q. Zhang, Integer flows and cycle covers of graphs, Marcel Dekker, Inc., New York Basel Hong Kong, 1997.


\end{thebibliography}

\end{document}